\documentclass[a4paper,12pt]{article} 

\usepackage{fourier} 
\usepackage[scaled=0.875]{helvet} 

\usepackage{titlesec}
\titlespacing\section{0pt}{30pt plus 4pt minus 2pt}{5pt plus 2pt minus 2pt}
\titlespacing\subsection{0pt}{15pt plus 4pt minus 2pt}{5pt plus 2pt minus 2pt}
\titlespacing\subsubsection{0pt}{12pt plus 4pt minus 2pt}{0pt plus 2pt minus 2pt}

\usepackage{tabularx}
\newcolumntype{L}[1]{>{\raggedright\arraybackslash}p{#1}}
\newcolumntype{C}[1]{>{\centering\arraybackslash}p{#1}}
\newcolumntype{R}[1]{>{\raggedleft\arraybackslash}p{#1}}

\usepackage{tikz}

\usepackage{paralist}                   

\usepackage{enumitem }
\setlist[itemize]{leftmargin=*} 

\usepackage{amsmath,amssymb,amsthm}
\usepackage{wasysym}
\usepackage{indentfirst}
\usepackage{graphicx}
\usepackage[header,title,titletoc]{appendix}
\usepackage{comment}
\usepackage[numbers]{natbib}             

\usepackage{geometry}
\usepackage[bookmarks,colorlinks]{hyperref} 
\hypersetup{colorlinks,%
           citecolor=blue,%
           filecolor=black,%
           linkcolor=blue,%
           urlcolor=blue}
\usepackage[stable]{footmisc} 

\renewcommand{\thefootnote}{\fnsymbol{footnote}}

\newcommand{\bd}[1]{{\bf #1}}

\usepackage{soul} 

\newcommand{\email}[1]{\href{mailto:#1}{\nolinkurl{#1}}}

\newtheorem{defi}{Definition}
\newtheorem{prop}{Proposition}

\newtheorem{ex}{Example}

\theoremstyle{definition}
\newtheorem{rem}{Remark}

\begin{document}

\title{On the Equivalence of Two Competing Affirmative Actions in School Choice}
\author{Yun Liu\footnotemark[1]}

\date{\today} 

\footnotetext[1]{Center for Economic Research, Shandong University, Jinan, 250100, China. Email: yunliucer@sdu.edu.cn.}

\maketitle
\setcounter{footnote}{0}
\renewcommand\thefootnote{\arabic{footnote}}

\thispagestyle{empty}

\begin{abstract}

This note analyzes the outcome equivalence conditions of two popular affirmative action policies, \emph{majority quota} and \emph{minority reserve}, under the student optimal stable mechanism. These two affirmative actions generate an identical matching outcome, if the market either is \emph{effectively competitive} or contains a sufficiently large number of schools.

\medskip

\noindent\bd{JEL Classification Number:} C78; D47; D78; I20 
\medskip

\noindent\bd{Keywords:} school choice; deferred acceptance; affirmative action; effective competition; large market.

\end{abstract}

\maketitle
\thispagestyle{empty}

\section{Introduction}
\label{sec:intro}

Affirmative action policies are meant to mitigate ethnic and social-economic disparities in schools through providing \emph{minority} students preferential treatments in the admission process. In practice, one popular policy design is the quota-based affirmative action (\emph{majority quota}, henceforth), which sets a maximum number less than a school's capacity to \emph{majority} students and leaves the difference to minority students. Based on \cite{GS62}'s deferred acceptance algorithm, \cite{AS03} shows that the \emph{student optimal stable mechanism} (SOSM) with majority quota is stable and strategy-proof for students. \cite{HYY} propose the reserve-based policy (\emph{minority reserve}, henceforth), which gives minority students preferential treatments up to the reserves. Their results indicate that without losing its stability and strategy-proofness, the minority reserve policy unilaterally outperforms its majority quota counterpart under the SOSM in term of students' welfare.

Given the prevalence of the majority quota policy in practice, introducing the minority reserve policy to the admission process could possibly evoke substantial political, administrative and cognitive costs to local communities which may offset or even surpass its theoretical efficiency edge. It is thus unnecessary to abandon the majority quota policy if we are able to perceive the causes which lead to the Pareto dominance relation between the majority quota and its minority reserve counterpart. In this note, we first introduce the \emph{effective competition} condition to assure that for each school with nonzero reserved seats for minorities, the number of minority students who list it as their first choice is no less than the number of its reservations, i.e., competition for each reserved seat is fierce among minority students. Our Proposition \ref{prop:equal} shows that the majority quota and the minority reserve generate the same outcome under the SOSM if the matching market is effectively competitive; in other words, the efficiency edge of the minority reserve policy over its majority quota counterpart essentially comes from the possible misallocation of reserved slots to less desirable schools with insufficient competition among minorities.

As our effective competition condition is not exclusive, there may have other subsets of \emph{finite} markets which are not effectively competitive but still generate identical matching outcomes under the two affirmative actions (see Example \ref{ex2} in Appendix \ref{app:example}), we further investigate the \emph{asymptotic} outcome equivalence condition of these two affirmative actions in a sequence of random markets of different sizes. Proposition \ref{prop:large} implies that the probability that the two affirmative actions generate the same matching outcome under the SOSM converges to one when the market contains sufficiently many schools with relatively few reserved seats. In other words, there is no need to distinguish these two affirmative action policies if the social planner can assure a sufficient supply of popular schools to the matching markets. 

Although our large market setting relies on a number of regularity conditions, this framework has been adopted in several recent analyses on the asymptotic properties of matching mechanisms in different contexts. Among others, \cite{IM05} claim that strategy-proofness is an \emph{approximate} Bayes-Nash equilibrium in two-sided one-to-one matching markets. \cite{KP09} extend this approximate strategy-proofness concept to the many-to-one matching markets. \cite{KPR13} and \cite{ABH14} prove the existence of asymptotically stable matching mechanism in the National Resident Matching Program with both single and married doctors; in particular, \cite{ABH14} improve the growth rate of couples at $n^{1-\varepsilon}$ from $\varepsilon \in (1/2,1)$ in \cite{KPR13} to $\varepsilon \in (0,1)$ by considering a particular sequence of proposals by couples, while preserving the linear growth of hospitals and single doctors. \cite{HKN16} reveal that all stable mechanisms asymptotically respect improvements of school quality (i.e., a school matches with a set of more desirable students if it becomes more preferred by students). \cite{Liu17aa} shows that the minority reserve policy is very unlikely to hurt any minority students in stable mechanisms. By alternately examining the vanishing of market disruptions from a fraction of agents in probability, \cite{CT19} suggest that the inefficiency of the deferred acceptance algorithm \citep{GS62} and the instability of the top trading cycle mechanism \citep{SS74} remain significant even when the market grows large.

\section{Model}						%
\label{sec:model}					%

\subsection{School Choice}
\label{sec:predef}

Let $S$ and $C$ be two finite and disjoint sets of students and schools, $|S| \ge 2$. There are two types of students, \emph{majority} and \emph{minority}. $S$ is partitioned into two subsets of students based on their types. Denote $S^M$ the set of majority students, and $S^m$ the set of minority students, $S = S^M \cup S^m$ and $S^M \cap S^m = \emptyset$. Each student $s \in S$ has a strict \emph{preference} order $P_s$ over the set of schools and being unmatched (denoted  by $s$). All students prefer to be matched with some school instead of herself, $c \,P_s\, s$, for all $s \in S$. Each school $c \in C$ has a total capacity of $q_c$ seats, $q_c \ge 1$, and a strict \emph{priority} order $\succ$ over the set of students which is complete, transitive, and antisymmetric. Student $s$ is unacceptable by a school if $e \succ_c s$, where $e$ represents an empty seat in school $c$.

A school choice \emph{market} is a tuple $G = (S, C, P, \succ, q)$, where $P = (P_i)_{i \in S}$, $\succ = (\succ)_{c \in C}$ and $q = (q_c)_{c \in C}$. Denote $P_{-i} = (P_j)_{j \in S \backslash i}$ and $\succ_{-c} = (\succ_{c'})_{c' \in C \backslash c}$. For a given $G$, assume that all components except $P$, are commonly known. Since the sets of schools and students are fixed, we simplify the market as $G = (P, \succ, q)$.

A \emph{matching} $\mu$ is a mapping from $S \cup C$ to the subsets of $S \cup C$ such that, for all $s \in S$ and $c \in C$: (i) $\mu(s) \in C \cup \{ s\}$; (ii) $\mu(s) = c$ if and only if $s \in \mu(c)$; (iii) $\mu(c) \subseteq S$ and $|\mu(c)| \le q_c$; and (iv) $|\mu(c) \cap S^M| \le q_c^M$. That is, a matching specifies the school where each student is assigned or matched with herself, and the set of students assigned to each school. A \emph{mechanism} $f$ is a function that produces a matching $f(G)$ for each market $G$.

\subsection{Affirmative Action Policies}		%
\label{sec:aap}						%

\cite{AS03} compose the \emph{student optimal stable mechanism with majority quota} algorithm (SOSM-Q henceforth), in which each school $c$ cannot admit more majority students than its type-specific \emph{majority quota} $q_c^M \le q_c$, $q^M = (q_c^M)_{c \in C}$, for all $c \in C$. A matching $\mu$ is \emph{blocked} by a pair of student $s$ and school $c$ with majority quota, if $c P_s \mu(s)$ and either $|\mu(c)| < q_c$ and $s$ is acceptable to $c$, or: (i) $s \in S^m$, $s \succ_c s'$, for some $s' \in \mu(c)$; (ii) $s \in S^M$ and $|\mu(c) \cap S^M| < q_c^M$, $s \succ_c s'$, for some $s' \in \mu(c)$; (iii) $s \in S^M$ and $|\mu(c) \cap S^M| = q_c^M$, $s \succ_c s'$, for some $s' \in \mu(c) \cap S^M$. A matching $\mu$ is \emph{stable with majority quota}, if $\mu(s) \,P_s\, s$ for all $s \in S$, and has no blocking pair.

The SOSM-Q algorithm runs as follows: 
\begin{itemize}

\item[] \textbf{Step 1}: Each student $s$ applies to her first-choice school (call it school $c$). The school $c$ rejects $s$ if either $q_c$ are filled by students who have higher priorities than $s$ at $c$, or $s \in S^M$ and $q_c^M$ are filled by majority students having higher priorities than $s$ at $c$. Each school $c$ tentatively accepts the remaining applicants until its capacity is filled while maintaining $|\mu(c)| \le q_c^M$, or the applicants are exhausted.

\item[] \textbf{Step k}: Each student $s$ who was rejected in Step $(k-1)$ applies to her next highest choice (call it school $c$, if any). Each school $c$ considers these students together with the applicants tentatively accepted from the previous steps. The school $c$ rejects $s$ if either $q_c$ are filled by students who have higher priorities than $s$ at $c$, or $s \in S^M$ and $q_c^M$ are filled by majority students having higher priorities than $s$ at $c$. Each school $c$ tentatively accepts the remaining applicants until its capacity is filled while maintaining $|\mu(c)| \le q_c^M$, or the applicants are exhausted.

\end{itemize}


The algorithm terminates either when every student is matched to a school or every unmatched student has been rejected by every acceptable school, which always terminates in a finite number of steps. Denote the new mechanism by $f^Q$, and its outcome in market $G^q$ by $f^Q (G^q)$, where $G^q = (P, \succ, (q, q^M))$. 

\medskip

\cite{HYY} further propose the more flexible \emph{student optimal stable mechanism with minority reserve} algorithm (SOSM-R henceforth), in which each school $c$ gives priority to minority applicants up to its \emph{minority reserve} $r_c^m \le q_c$, $r^m = (r_c^m)_{c \in C}$, and allows to accept majority students up to its capacity if there are not enough minority applicants to fill the reserves. A matching $\mu$ is \emph{blocked} by a pair of student $s$ and school $c$ with minority reserve, if $s$ strictly prefers $c$ to $\mu(s)$ and either $|\mu(c)| < q_c$ and $s$ is acceptable to $c$, or: (i) $s \in S^m$, $c$ strictly prefers $s$ to some $s' \in \mu(c)$; (ii) $s \in S^M$ and $|\mu(c) \cap S^m| > r_c^m$, $c$ strictly prefers $s$ to some $s' \in \mu(c)$; (iii) $s \in S^M$ and $|\mu(c) \cap S^m| \le r_c^m$, $c$ strictly prefers $s$ to some $s' \in \mu(c) \cap S^M$. A matching $\mu$ is \emph{stable with minority reserve}, if $\mu(s) \,P_s\, s$ for all $s \in S$, and has no blocking pair.

The SOSM-R algorithm runs as follows: 
\begin{itemize}
\item[] \textbf{Step 1:} Each student $s$ applies to her first-choice school. Each school $c$ first tentatively accepts up to $r_c^m$ minorities with the highest priorities if there are enough minority applicants; it then tentatively accepts the remaining applicants with the highest priorities until its capacity is filled or the applicants are exhausted. The rest (if any) are rejected.


\item[] \textbf{Step k:} Each student $s$ who was rejected in Step $(k-1)$ applies to her next highest choice (if any). Each school $c$ considers these students together with the applicants tentatively accepted from the previous steps. $c$ first tentatively accepts up to $r_c^m$ minorities with the highest priorities if there are enough minority applicants; it then tentatively accepts the remaining applicants with the highest priorities until its capacity is filled or the applicants are exhausted. The rest (if any) are rejected.

\end{itemize}


The algorithm terminates either when every student is matched to a school or every unmatched student has been rejected by every acceptable school, which always terminates in a finite number of steps. Denote the new mechanism by $f^R$, and its outcome in market $G^r$ by $f^R (G^r)$, where $G^r = (P, \succ, (q, r^m))$.

Denote $\Gamma= (P, \succ, (q^M, r^m))$ the market when we compare the effects of a majority quota policy in market $G^q$ and its corresponding minority reserve policy in market $G^r$, where $r_c^m + q_c^M =q_c$, $\forall c \in C$.

\section{Results} 		%
\label{sec:equal}															%

\subsection{Outcome Equivalence in Finite Markets}    
\label{ssec:finite}                          

The efficiency loss of the majority quota policy, as indicated by \cite{HYY},  essentially comes from a rejection chain initiated by a school with excessive majority applicants over its quota and insufficient number of minority applicants. We first introduce the following condition to guarantee sufficient competition among minority students for each reserved seat.

\begin{defi} \label{def_ec}

Consider majority quotas $q^M$ and minority reserves $r^m$ such that $r^m + q^M =q$. Let $\bar{S}_c^m \subseteq S^m$ be the set of minority students who rank school $c$ as their first choices, i.e., $c P_s c'$, $\forall\, s \in \bar{S}_c^m, c' \in C \backslash c$. A market $\Gamma$ is \emph{effectively competitive}, if for each $c \in C$ with $q_c \ge r_c^m > 0$ (or equivalently, $q_c > q_c^M \ge 0$), we have $|\bar{S}_c^m| \ge r_c^m$, where $r_c^m = q_c - q_c^M$.

\end{defi}

In words, we say an affirmative action is effectively implemented in a school choice market if each school with nonzero reservations has at least as many minority applicants as its reservations in the first step of the matching process.

\begin{prop} \label{prop:equal}
Consider majority quotas $q^M$ and minority reserves $r^m$ such that $r^m + q^M =q$. $f^Q (\Gamma) = f^R (\Gamma)$, if ~$\Gamma$ is effectively competitive.
\end{prop}
\begin{proof}
See Appendix \ref{app:pprop1}.
\end{proof}

Proposition \ref{prop:equal} indicates that if the \emph{reserved} seats (i.e., $r_c^m > 0$ under the minority reserve policy or $q_c - r_c^m \ge 0$ under its majority quota counterpart, $\forall c \in C$) are properly allocated among those highly desired schools by the minorities, the two affirmative actions will produce an identical matching outcome as the same set of students is tentatively accepted in each step of the SOSM with either majority quotas or minority reserves. In other words, the efficiency edge of the minority reserve policy over its majority quota counterpart essentially comes from those reserved seats with insufficient competition among minorities.

\begin{rem} \label{rem:strategy}

Note that altering from one affirmative action policy to the other is not a mere change of the matching rule (e.g., replacing the existing Boston school choice mechanism with the SOSM in Boston city), it also imposes different priority orders across majority and minority students in schools with excessive majority applicants and insufficient number of minority applicants (see Expression \eqref{eq:spiltq} and \eqref{eq:spiltr} in Appendix \ref{app:pprop1}). Although our effective competition condition is not exclusive---there are other subsets of markets which are not effectively competitive but still generate identical matching outcomes under these two affirmative actions (see Example \ref{ex2} in Appendix \ref{app:example}), it guarantees an identical alteration on the priorities of each school with nonzero reserved seats. Since a large portion of performance comparison criteria is only applicable to an identical underlying market under different matching mechanisms, the effective competition condition essentially characterizes a subset of markets in which we are able to compare the performance of these two affirmative action policies.

\end{rem}

\subsection{Asymptotic Equivalence in Large Markets}    
\label{ssec:large}                          

Motivated by the success of applying asymptotic analysis technique to restore some negative results in finite matching markets \citep{IM05,KP09,KPR13,ABH14,HKN16,CT19}, we further explore the \emph{asymptotic} outcome equivalence condition of these two affirmative actions in a sequence of random markets of different sizes.

Define a \emph{random market} as a tuple $\tilde{\Gamma} = ((S^M, S^m), C, \succ, (q^M, r^m), k, (\mathcal{A},\mathcal{B}))$,  where $k$ is a positive integer, $\mathcal{A} = (\alpha_c)_{c \in C}$ and $\mathcal{B} = (\beta_c)_{c \in C}$ are the respective probability distributions on $C$, with $\alpha_c, \beta_c > 0$ for each $c \in C$. We assume that $\mathcal{A}$ for majorities to be different from $\mathcal{B}$ for minorities to reflect their distinct favors for schools. Each random market induces a market by randomly generated preference orders of each student $s$ according to the following procedure introduced by \cite{IM05}:
\begin{itemize}

\item[] \textbf{Step 1:} Select a school independently from the distribution $\mathcal{A}$ (resp. $\mathcal{B}$). List this school as the top ranked school of a majority student $s \in S^M$ (resp. minority student $s \in S^m$).
\item[] \textbf{Step $\mathbf{l \le k}$:}  Select a school independently from $\mathcal{A}$ (resp. $\mathcal{B}$) which has not been drawn from steps 1 to step $l-1$. List this school as the $l^{th}$ most preferred school of a majority student $s \in S^M$ (resp. minority student $s \in S^m$).

Each majority (resp. minority) student finds these $k$ schools acceptable, and only lists these $k$ schools in her preference order.

\end{itemize}


A sequence of \emph{random markets} is denoted by $(\tilde{\Gamma}^1,\tilde{\Gamma}^2,\dots)$, where $\tilde{\Gamma}^n = ((S^{M,n}, S^{m,n}), C^n, \succ_n, (q^{M,n}, r^{m,n}), k^n, (\mathcal{A}^n, \mathcal{B}^n))$ is a random market of size $n$, with $|C^n| = n$ as the number of schools and $|r^{m,n}|$ the number of seats reserved for minorities. We introduce the following regularity conditions to guarantee the convergence of the random markets sequence.

\begin{defi} \label{def:regular}

Consider majority quotas $q^M$ and minority reserves $r^m$ such that $r^m + q^M =q$. A sequence of random markets $(\tilde{\Gamma}^{1},\tilde{\Gamma}^{2},\dots)$ is \emph{regular}, if there exist $a \in [0, \frac{1}{2})$, $\lambda, \kappa, \theta > 0$, $r \ge 1$, and positive integers $k$ and $\bar{q}$, such that for all $n$:
\begin{enumerate}[leftmargin=*]
\item $k^n \le k$;
\item $q_c \le \bar{q}$ for all $c \in C^n$;
\item $|S^n| \le \lambda n$, $\sum_{c \in C} q_c - |S^n| \ge \kappa n$;
\item $|r^{m,n}| \le \theta n^a$; 
\item $\frac{\alpha_c}{\alpha_{c'}} \in [\frac{1}{r}, r]$, $\frac{\beta_c}{\beta_{c'}} \in [\frac{1}{r}, r]$, for all $c, c' \in C^n$;
\item $\alpha_c=0$, for all $c \in C^n$ with $q_c^M = 0$.
\end{enumerate}
\end{defi}

Condition (1) assumes that the length of students' preferences is bounded from above across schools and markets. This is motivated by the fact that students' reported preference orders observed in many practical school choice programs are quite short; for example, about three quarters of students ranked less than 12 schools among over 500 school programs in New York City \citep{AS09}, whereas only less than 10\% of students rank more than 5 schools at the elementary school level out of around 30 different schools in their own walk-zone schools in Boston \citep{AS06}. 

Condition (2) requires that the number of seats in any school is also bounded across schools and markets; that is, even some schools tend to enroll more students than others, the difference of their capacities is limited. 

Condition (3) requires that the number of students does not grow much faster than the number of schools; in addition, there is an excess supply of school capacities to accommodate all students, which is consistent with most public school choice programs in practice. Note that Condition (3) does not distinguish the growth rate between majority and minority students, because minority students are generically treated as the intended beneficial student groups from affirmative action policies rather than race or other single social-economic status; thus, the number of minority students is not necessarily less than majorities.

Condition (4) requires that the number of seats reserved for minority students grows at a slower rate of $O(n^a)$, where $a \in [0, \frac{1}{2})$. This regularity condition guarantees that any market disruption caused by schools with either a majority quota or its minority reserve counterpart is likely to be absorbed by other schools without affirmative actions when the market contains sufficiently many schools. We examine the alternative regularity conditions with a slower growth of the number of minority students in the Online Appendix, see: \url{http://yunliueconomics.weebly.com/uploads/3/2/2/1/32213417/equal_aa_app.pdf}.

Condition (5) is termed as \emph{moderate similarity} in \cite{HKN16}, which is also called sufficient thickness in \cite{KP09,KPR13} or uniformly bounded preferences in \cite{ABH14}. It requires that the popularity of different schools, as measured by the probability of being selected by students from $\mathcal{A}$ for majorities and $\mathcal{B}$ for minorities, does not vary too much; in other words, the popularity ratio of the most favorable school to the least favorable school is bounded. 

Condition (6) requires that a majority student will not select a school that can only accept minority students after imposing the quota $q_c^M = 0$, as these two affirmative actions trivially induce disparate matching outcomes in any arbitrarily large markets when a majority student applies to a school with zero majority quota. We employ the probability distributions $\mathcal{A}$ and $\mathcal{B}$ for majority and minority students to illustrate their distinct preferences over schools, especially the exclusion of schools with zero majority quota for majority students as required here.

\begin{defi} \label{def_large}
For any random market $\tilde{\Gamma}$, let $\eta_c(\tilde{\Gamma};f,f')$ be the probability that $f(\tilde{\Gamma}) \ne f'(\tilde{\Gamma})$. We say two mechanisms are \emph{outcome equivalence in large markets}, if for any sequence of random markets $(\tilde{\Gamma}^1,\tilde{\Gamma}^2,\dots)$ that is regular, $\max_{c \in C^n} \eta_c(\tilde{\Gamma}^n; f,f') \to 0$, as $n \to \infty$; that is, for any $\varepsilon > 0$, there exists an integer $m$ such that for any random market $\tilde{\Gamma}^n$ in the sequence with $n>m$ and any $c \in C^n$, we have $\max_{c \in C^n} \eta_c(\tilde{\Gamma}^n;f,f') < \varepsilon$.

\end{defi}

We are now ready to present our main argument on the asymptotic outcome equivalence of the majority quota policy and its minority reserve counterpart under the SOSM for a regular sequence of random markets.

\begin{prop} \label{prop:large}

The SOSM-Q and its corresponding SOSM-R are outcome equivalence in large markets.

\end{prop}

\begin{proof}
See Appendix \ref{app:pprop2}.
\end{proof}

As the two affirmative actions will result in different matching outcomes only when some schools with nonzero \emph{reserved} seats  (i.e., $r_c^{m,n}$ under the minority reserve policy and $q_c^n - q_c^{M,n}$ under its corresponding majority quota policy for some $c \in C^n$) have excessive majority applicants and insufficient number of minority applicants (Proposition \ref{prop:equal}), Proposition \ref{prop:large} implies that when the reserved seats is not growing very fast in the sequence of random markets, it is very unlikely for any two particular students apply to the same school with nonzero reserved seats under either the SOSM-Q or the SOSM-R when the market contains sufficiently many schools.

\begin{rem} \label{rem:reg}
One key regularity condition in Definition \ref{def:regular} is that we assume the growth rate of reserved seats is lower than $\sqrt{n}$. To see why our asymptotic outcome equivalence result would not hold if the number of reserved seats $|r^{m,n}|$ grows at $n^a$, $a \in [1/2,1]$, as $n$ approaches infinity, let us consider a random market $\tilde{\Gamma}^n$ without any reserved seats at first, which clearly generates an identical stable matching $\mu^n$ under either the SOSM-Q or the SOSM-R as no schools' priorities are affected by the affirmative actions. We then add reserved seats one at a time into $\tilde{\Gamma}^n$ to measure its effects on $\mu^n$ when the reserved seat is either reserved for minorities (i.e., minority reserve) or treated as the admission cap for majority students (i.e., majority quota) in a school $c$. The first reserved seat $r_1 \in r^{m,n}$ will initiate a rejection chain under the SOSM-Q but not the SOSM-R if the reserved seat is allocated to a school $c$ that has already accepted more majority students than its majority quota, $|\mu^n(c) \cap S^{M,n}| > q_c^{M,n}$, as $c$'s least favorable majority mate (denoted by $s_1$) will not be rejected under the SOSM-R given its flexible admission cap for majorities; however, $c$ is forced to reject $s_1$ under its rigid majority quota of the SOSM-Q, which makes the rejected majority student $s_1$ apply to her next favorable school. The rejection chain can cause several students (either majorities or minorities) who were temporarily assigned to some schools continue applying. Since a school $c' \in C \backslash c$ will not reject any student if it has a vacant position, the rejection chain terminates when a student rejected from her previously matched school accepted by a school with a vacancy. As Condition (3) ensures the excess supply of school capacities, the probability that $c'$ has a vacancy is $1 - \lambda$ (for simplicity we assume that $\lambda \in (0,1)$, $q_c=1$ for all $c \in C^n$, while $\mathcal{A}$ and $\mathcal{B}$ are both uniformly distributed). Following a similar procedure in \cite{ABH14}, we can show that with probability $1- 1/n$ the number of schools involved in the rejection chain initiated by a single reserved seat is upper bounded by $\lambda \log n/ (1-\lambda)$. When the second reserved seat $r_2$ is added to the market, it can also evict matched students from at most $\lambda \log n/ (1-\lambda)$ schools, among which the probability that it can affect the schools with the two reserved seats $r_1$ and $r_2$ is upper bounded by
\[
\frac{\lambda \log n}{1-\lambda} \cdot \frac{2}{n}.
\]
Since we have at most $\bar{R} = \theta n^a$ reserved seats, the probability that any school is involved in the rejection chain with $\bar{R}$ reserved seats is 
\[
\sum_{r=1}^{\bar{R}} \frac{r \lambda \log n}{(1-\lambda)n} \;\ < \; \frac{\bar{R}^2 \lambda \log n}{(1-\lambda)n} \; = \; \frac{\theta^2 n^{2a} \lambda \log n}{(1-\lambda)n} \;=\; O \left( \frac{\log n}{n^{1-2a}}\right),
\]
which converges to zero as $n$ goes to infinity. Clearly, the outcomes of the SOSM-Q and its corresponding SOSM-R will not be asymptotically equivalent in large markets for any $a \ge 1/2$, as the probability that a school with reserved seats to be involved in the rejection chains initiated by the rejection of a majority student under the SOSM-Q but not under the SOSM-R will not converge to zero as illustrated above. In other words, our argument is more close to the direct rejection approach used in \cite{KPR13} which prohibits any rejections of couples from their currently matched hospitals with a high probability, as a different application order as in \cite{ABH14} will not change the set of students a school matched with under either the SOSM-Q or the SOSM-R given the deferred acceptance nature of the SOSM based on the Gale-Shapley's original algorithm.

\end{rem}

\begin{rem} \label{rem:length}

Similar to the arguments in \cite{KP09} (see their Footnote 32) and \cite{HKN16}, we can relax the length of preference orders (Condition (1) of Definition \ref{def:regular}) to $k^n=o(\log(n))$, i.e., the number of schools that are acceptable to each student grows without bound but at a sufficiently slow rate. We preserve Condition (1) because: (i) the main mechanics of our large market model and results are robust to the changes of preference length as long as the slower growth of reserved seats (of Condition (4)) and the moderate similarity (of Condition (5)) are satisfied; (ii) assuming $k^n \le k$ also complies with the observation that most reported preference orders in real world are quite short---ranking many schools to form a lengthy preference order is (physically and mentally) costly for most students.

\end{rem}

\begin{rem} \label{rem:hetero}

We can also consider other alternative assumptions of our large market model, for example: (i) students can be assigned into groups (geographic regions) with heterogeneous group-specific preference distributions, i.e., $\mathcal{A}_{\hat{s}} \ne \mathcal{A}_{\hat{s}'}$ and $\mathcal{B}_{\tilde{s}} \ne \mathcal{B}_{\tilde{s}'}$, for some $\hat{s}, \hat{s}' \in S^M$ and $\tilde{s}, \tilde{s}' \in S^m$; (ii) not all but only a sufficiently large subset of schools satisfies the moderate similarity of Condition (5); (iii) even though each student independently draw her preferences from either $\mathcal{A}$ or $\mathcal{B}$, students preferences are allowed to remain certain correlations in the sense that their preferences are correlated through a random state variable $\sigma$, but they are still conditionally independent of $\sigma$ (e.g., $\sigma$ can represent the changes of teaching quality across schools). Provided that our Condition (4) is satisfied, these alternative assumptions will not invalidate our asymptotic outcome equivalence result (see discussions in \citep{KP09,KPR13,HKN16} under similar large matching market settings, especially Section 2.4 of \citep{HKN16}'s online appendix).


\end{rem}

\section{Concluding Remarks}
\label{sec:conclusion}

This note studies the outcome equivalence conditions of the majority quota policy and the minority reserve policy under the student optimal stable mechanism. Our results imply that the same set of students are matched under these two affirmative actions if the social planner can either properly assign each reserved slot to highly competitive schools among minority students, or promise a sufficient supply of desirable schools for both the minority and majority students. Given the transparent priority orders and historical preferences data in many practical school choice problems, a more prudent allocation of the reserved seats is clearly much more cost-effective to accommodate affirmative action policies compared to introducing alternative matching mechanisms to local communities. Therefore, apart from redesigning new matching mechanisms, future research can also work on identifying the corresponding (asymptotic) equivalence conditions of different affirmative actions in other conventional matching mechanisms.





\appendix           
\addappheadtotoc

\section{Appendix}
\label{app}

\subsection{Examples} 
\label{app:example}

\begin{ex} \label{ex1} \normalfont

\emph{(The SOSM-Q and the SOSM-R produce different matching outcomes in an ineffectively competitive market.)} Consider the following market $\Gamma = (P, \succ, (q^M, r^m))$ with two schools $C= \{c_1,c_2\}$, and four students $S = \{s_1, s_2, s_3, s_4\}$ where $S^M = \{s_1, s_3\}$ and $S^m = \{s_2, s_4\}$. $q_{c_1} = 2$ and $q_{c_2} = 2$. Schools and students have the following priority and preference orders:
\begin{table}[ht]
\centering
\begin{tabular}{C{1.5cm}| C{1.5cm}}
$\succ_{c_k,\, k =1,2}$ & $P_{s_i, \; i =1,2,3,4}$ \\ [0.5ex] 
\hline
$s_1$ & $c_1$    \\
$s_2$ & $c_2$    \\
$s_3$ &          \\
$s_4$ &          \\ [0.5ex]
\end{tabular}
\end{table}

Suppose that $\Gamma$ has the following majority quota and its corresponding minority reserve: $(q_{c_1}^M,q_{c_2}^M) = (1,0)$, or correspondingly, $(r_{c_1}^m,r_{c_2}^m) =(1,2)$ (i.e., no majority student is acceptable in $c_2$). $\Gamma$ is ineffectively competitive as no minority student applies to $c_2$ at the first step of the two mechanisms (i.e., no minorities list $c_2$ as their first choice). The SOSM-Q and the SOSM-R produce different matching outcomes as:
\begin{equation*}
f^{Q} (\Gamma) = \left(
\begin{array}{cc}
c_1 & c_2  \\
\{s_1, s_2\} & s_4
\end{array} \right)\qquad\qquad
f^{R} (\Gamma) = \left(
\begin{array}{cc}
c_1 & c_2  \\
\{s_1, s_2\} & \{s_3, s_4\}
\end{array} \right)
\end{equation*}
which leave $s_3$ unmatched under SOSM-Q.

\end{ex} 

\smallskip

\begin{ex} \label{ex2} \normalfont

\emph{(The SOSM-Q and the SOSM-R produce an identical matching outcome in an ineffectively competitive market.)} Consider the following market $\Gamma = (P, \succ, (q^M, r^m))$ with two schools $C= \{c_1,c_2\}$, and three students $S = \{s_1, s_2, s_3\}$ where $S^m = \{s_1\}$ and $S^M = \{s_2, s_3\}$. $q_{c_1} = 3$ and $q_{c_2} = 2$. Schools and students have the following priority and preference orders:
\begin{table}[ht]
\centering
\begin{tabular}{C{1cm} C{1cm} | C{1cm} C{1cm}}
$\succ_{c_1}$ & $\succ_{c_2}$ &  $P_{s_i, \; i =1, 3}$ & $P_{s_2}$ \\ [0.5ex] 
\hline
$s_1$ & $s_3$ & $c_1$ & $c_2$  \\
$s_2$ & $s_2$ & $c_2$ & $c_1$  \\
$s_3$         & $s_1$ &       &        \\ [0.5ex]
\end{tabular}
\end{table}

Suppose that $\Gamma$ has the following majority quota and its corresponding minority reserve: $(q_{c_1}^M,q_{c_2}^M) = (1,1)$, or correspondingly, $(r_{c_1}^m,r_{c_2}^m) =(2,1)$. The SOSM-Q and the SOSM-R produce the same matching outcome:
\begin{equation*}
f^{Q} (\Gamma) = f^{R} (\Gamma) = \left(
\begin{array}{cc}
c_1 & c_2  \\
\{s_1, s_3\} & s_2
\end{array} \right)
\end{equation*}

\end{ex} 

\subsection{Proof of Proposition \ref{prop:equal}} 
\label{app:pprop1}

Consider a market $\Gamma$ with either the majority quota $q^M$ or its corresponding minority reserve $r^m$, such that $r^m + q^M =q$. Given the majority quota policy $q^M$, we can split each school $c$ with capacity $q_c$ and a majority quota $q_c^M$ into two corresponding sub-schools, the \emph{original sub-school} ($c^o$) and the \emph{quota sub-school} ($c^q$), $c = (c^o, c^q)$. $c^o$ has a capacity of $q_c^M$ and maintains the original priority order $\succ$, whereas $c^q$ has a capacity of $q_c - q_c^M$ and its new priority $\succ_{c}^q$ is 
\begin{equation} \label{eq:spiltq}
\succ_{c}^q \; \equiv
\begin{cases}
s \succ_c s' \quad & \text{if} \quad s, s' \in S^m\\
e \succ_c^q s \quad & \text{if} \quad s \in S^M
\end{cases}
\end{equation}
i.e., $c^q$ keeps the same pointwise priority orders as school $c$ for all minority students, but prefers leaving an empty seat ($e$) to accepting any majority student.

Correspondingly, when $\Gamma$ has the minority reserve, $r^m =q - q^M$, we can split each school $c$ with capacity $q_c$ and a minority reserve $r_c^m$ as two corresponding sub-schools, the \emph{unaffected sub-school} ($c^u$) and the \emph{reserve sub-school} ($c^r$), $c = (c^u, c^r)$. $c^u$ has a capacity of $q_c - r_c^m$ and maintains the original priority order $\succ$. $c^r$ has a capacity of $r_c^m$ and its new priority $\succ_{c}^r$ is
\begin{equation} \label{eq:spiltr}
\succ_{c}^r \; \equiv
\begin{cases}
s \succ_c s' \quad & \text{if} \quad s, s' \in S^m\\
s \succ_c s' \quad & \text{if} \quad s, s' \in S^M\\
s \succ_c^r s' \quad & \text{if} \quad s \in S^m, \, s' \in S^M
\end{cases}
\end{equation}
i.e., $c^r$ keeps the same pointwise priority orders as school $c$ for all majority students and all minority students, but it prefers all minorities to any majorities.

When the market $\Gamma$ is effectively competitive, all the \emph{reserved} seats (i.e., $|r^m|$ under the minority reserve policy and $|q - q^M|$ under its corresponding majority quota policy) are filled by minority students at Step 1 of the SOSM-Q and the SOSM-R, which ensures no majority student can replace any minority students who have tentatively filled these reserved seats in later steps. Thus, we have $\succ_{c}^q = \succ_{c}^r$, for each $c \in C$, as illustrated above. 

Since different affirmative action policies will not alter students' preference orders by assumption, each school will receive the same set of applicants at  Step 1 of the SOSM-Q and the SOSM-R. Together with $\succ_{c}^o \equiv \succ_{c}^u$ and $\succ_{c}^q = \succ_{c}^r$, for all $c \in C$, we know that each school also accepts the same set of students at Step 1 of the SOSM-Q and the SOSM-R. 

Denote $\mu^Q_k(c)$ (resp. $\mu^R_k(c)$) the set of students accepted by school $c$ at Step $k$ of the SOSM-Q (resp. SOSM-R), $k \ge 1$. Assume that each school receives and tentatively accepts an identical set of students until Step $k$ of the SOSM-Q and the SOSM-R, $\mu^R_k(c) = \mu^Q_k(c)$. We will argue by contradiction.

\emph{Case (i):} at Step $k+1$, let $s$ be a student who was rejected by school $c$ under the SOSM-Q but was tentatively accepted by $c$ under the SOSM-R, i.e., $s \in \mu^R_{k+1}(c) \backslash \mu^Q_{k+1}(c)$. As $\mu^R_k(c) = \mu^Q_k(c)$, while $c$ receives the same set of applicants at the beginning of Step $k+1$ of the SOSM-Q and the SOSM-R, we have $|\mu^R_k(c)|= |\mu^Q_k(c)|=q_c$ (i.e., $c$ has no empty seat left at the beginning of Step $k+1$ of either the SOSM-Q or the SOSM-R); otherwise, $s$ will not be rejected from $c$ under the SOSM-Q given that $\mu^Q_k(c^q)= \mu^R_k(c^r) \in S^m$. Therefore, there is another student $s' \in \mu^Q_{k+1}(c) \backslash \mu^R_{k+1}(c)$, and $s' \succ_{c^x} s$, where $x = q$ if $s\in S^m$, and $x = o$ if $s\in S$, under the SOSM-Q. Given that $\succ_{c}^o \equiv \succ_{c}^u$ and $\succ_{c}^q = \succ_{c}^r$, $\forall c \in C$, we have  $s' \succ_{c^x} s$, where $x = r$ if $s\in S^m$, and $x = u$ if $s\in S$, under the SOSM-R. $(c, s')$ clearly forms a blocking pair which violates the stability of the SOSM-R. This contradicts the assumption that school $c$ accepts different sets of students at Step $k+1$ of the SOSM-Q and the SOSM-R.

\emph{Case (ii):} at Step $k+1$, let $s$ be a student who was rejected by school $c$ under the SOSM-R but was tentatively accepted by $c$ under the SOSM-Q, i.e., $s \in \mu^Q_{k+1}(c) \backslash \mu^R_{k+1}(c)$. Similarly, we can see there is another student $s' \in \mu^R_{k+1}(c) \backslash \mu^Q_{k+1}(c)$, and $s' \succ_{c^x} s$, where $x = r$ if $s\in S^m$, and $x = u$ if $s\in S$, under the SOSM-R. Given that $\succ_{c}^o \equiv \succ_{c}^u$ and $\succ_{c}^q = \succ_{c}^r$, $\forall c \in C$, we have $s' \succ_{c^x} s$, where $x = q$ if $s\in S^m$, and $x = o$ if $s\in S$, under the SOSM-Q. Thus, $(c, s')$ forms a blocking pair which violates the stability of the SOSM-Q. We get a contradiction.

\subsection{Proof of Proposition \ref{prop:large}} 
\label{app:pprop2}

We first incorporate a stochastic variant of the SOSM introduced by \cite{MW71} with affirmative actions in the sense that all majority and minority students are ordered in some predetermined (for instance, random) manner. Following the approach used in \cite{IM05,KP09,KPR13,HKN16,CT19}, we use the \emph{principle of deferred decisions} and the technique of \emph{amnesia}, which are originally proposed by \cite{K76,KMP90}, to simplify the stochastic process of Algorithm 1. By the principle of deferred decisions, we assume that students do not know their preferences in advance and whenever a student has an opportunity to submit applications according to the predetermined order, she applies to her most preferred school among those that has not yet rejected her. Since a school's acceptances could depend on the set of students that have applied to this school, we also apply the technique of amnesia to solve such dependency (from the past applications) in the sense that each student makes her applications randomly to a school in $C$ as she cannot remember any of the schools she has previously applied to. Assuming students are \emph{amnesiac} does not affect the final outcome of the algorithm, since if a student applies to a school that has already rejected her, she will be rejected again.

To ease the superscript notations, we relabel the set of majority students $S^M$ and the minority students $S^m$ by $S(M)$ and $S(m)$, respectively. Let $A_s$ and $B_s$ be the respective sets of schools that a majority  (resp. minority) student has already drawn from $\mathcal{A}^n$ and $\mathcal{B}^n$. When $|A_s| = k$ (resp. $|B_s| = k$) is reached, $A_s$ (resp. $B_s$) is the set of schools that is acceptable to the majority student $s \in S(M)$ (resp. minority student $s \in S(m)$).

\paragraph{\textbf{Algorithm 1.}} Stochastic SOSM with Affirmative Actions \label{algo1}
\begin{enumerate}

\item Initialization: Let $l(M)=1$ and $l(m)=1$. For every $s \in S(M)$ (resp. $s \in S(m)$), let $A_s = \emptyset$ (resp. $B_s = \emptyset$). Since Algorithm 1 preserves the tentative acceptance property of the Gale-Shapley's algorithm, its outcome is independent of the order of applications made by students. Order all the students in $S$ in an arbitrarily fixed manner. \label{sec:l0}



    

\item Choosing the applicant: \label{sec:a2}

    \begin{enumerate}
    
    \item Choose a student $s \in S$ according to the students' application order of Step \eqref{sec:l0}, if  $l(M) \le |S(M)|$ or $l(m) \le |S(m)|$. 

    \begin{enumerate}
        
        \item If $s \in S(M)$, then let $s$ be the $l(M)^{th}$ majority student, and increment $l(M)$ by one. 

        \item  If $s \in S(m)$, then let $s$ be the $l(m)^{th}$ minority student, and increment $l(m)$ by one. 
    \end{enumerate}

    \item If not, then terminate the algorithm. \label{sec:a2b}
        
    \end{enumerate}
    
\item Choosing the application: \label{sec:a3}
    
    \begin{enumerate}
    \item If $s \in S(M)$:

    \begin{enumerate}
  
    \item If $|A_s| \ge k$, as she has already applied to every acceptable school, then go back to Step \eqref{sec:a2}.

    \item If not, select $c$ randomly from the distribution $\mathcal{A}^n$ until $c \notin A_s$. Add $c$ to $A_s$, if $c$ has not rejected $s$ yet. \label{sec:a3a}

    
    \end{enumerate}

    \item If $s \in S(m)$:

    \begin{enumerate}
  
    \item If $|B_s| \ge k$, as she has already applied to every acceptable school, then go back to Step \eqref{sec:a2}.
    
    \item If not, select $c$ randomly from the distribution $\mathcal{B}^n$ until $c \notin B_s$. Add $c$ to $B_s$, if $c$ has not rejected $s$ yet.  \label{sec:a3b}
    
    \end{enumerate}
    \end{enumerate}
    
\item Acceptance and/or rejection: 

    \begin{enumerate}

    \item If $c$ has neither a majority quota nor a minority reserve:

    \begin{enumerate}
        \item[i.] If $c$ has an empty seat, then $s$ is tentatively accepted. Go back to Step \eqref{sec:a2}.

        \item[ii.] If $c$ has no empty seat and prefers each of its current mates to $s$, $c$ rejects $s$. Go back to Step \eqref{sec:a3}.

        \item[iii.] If $c$ has no empty seat but it prefers $s$ to one of its current mates, then $c$ rejects the matched student with the lowest priority. Let this rejected student be $s$ and go back to Step \eqref{sec:a3}.
     \end{enumerate} 

     \item If $c$ has either a majority quota or a minority reserve, but has not received any application yet, then $s$ is tentatively accepted as no majority students will apply to a school with $q_c^M =0$ (Condition (6) of Definition \ref{def:regular}). Go back to Step \eqref{sec:a2}.

    \item If $c$ has a majority quota, and has tentatively accepted some students: \label{sec:a4c}

        \begin{enumerate}
          \item[i.] If $c$ has an empty seat:

          \begin{enumerate}
                \item[(A)] If $s \in S(M)$ and $|\mu(c) \cap S(M)| < q_c^M$, $s$ is tentatively accepted. Go back to Step \eqref{sec:a2}.

                \item[(B)] If $s \in S(M)$ and $|\mu(c) \cap S(M)| = q_c^M$:
                \begin{itemize}
                    \item[-] If $c$ prefers each of its current matched majority students to $s$, then $c$ rejects $s$. Go back to Step \eqref{sec:a3}.

                    \item[-] If $c$ prefers $s$ to one of its current matched majority students, $c$ rejects the matched majority student with the lowest priority. Let this rejected majority student be $s$ and go back to Step \eqref{sec:a3}.
                \end{itemize}

                \item[(C)] If $s \in S(m)$, then $s$ is tentatively accepted. Go back to Step \eqref{sec:a2}.
              \end{enumerate}

          \item[ii.] If $c$ has no empty seat and prefers each of its current mates to $s$, $c$ rejects $s$. Go back to Step \eqref{sec:a3}.

          \item[iii.] If $c$ has no empty seat but it prefers $s$ to one of its current mates, then $c$ rejects the matched student with the lowest priority while not admitting more majority students than its majority quota $q_c^M$. Let this rejected student be $s$ and go back to Step \eqref{sec:a3}.

        \end{enumerate}

         \item If $c$ has a minority reserve, and has tentatively accepted some students:  \label{sec:a4d}

        \begin{enumerate}
          \item[i.] If $c$ has an empty seat, then $s$ is tentatively accepted. Go back to Step \eqref{sec:a2}.

          \item[ii.] If $c$ has no empty seat and prefers each of its current mates to $s$, $c$ rejects $s$. Go back to Step \eqref{sec:a3}.

          \item[iii.] If $c$ has no empty seat but it prefers $s$ to one of its current mates, then $c$ rejects the matched student with the lowest priority. Let this rejected majority student be $s$ and go back to Step \eqref{sec:a3}.

        \end{enumerate}

    \end{enumerate}

\end{enumerate}

The stochastic SOSM with affirmative actions algorithm terminates at Step \eqref{sec:a2b}. By the principle of deferred decisions, the probability that Algorithm 1 arrives at any steps is identical regardless of whether the random preferences are drawn at once in the beginning or drawn iteratively during the execution of the algorithm.


Let $\bar{R} = \theta n^a$ be the upper bound on the number of seats reserved for minority students in the random market $\tilde{\Gamma}^n$. Also, let $ \text{Prob} \, ( \tilde{\Gamma}^{n} )$ be the probability that the two affirmative actions generate an identical matching outcome in a random market $\tilde{\Gamma}^n$. Our next step is to argue that when the market is sufficiently large while the number of reserved seats is not growing very fast, the majority quota policy and its corresponding minority reserve policy will produce the same matching outcome under Algorithm 1 with a high probability.

Denote $C^r$ the set of schools with nonzero reserved seats (i.e., $r_c^m>0$ under the minority reserve  policy and $q_c^M = q_c - r_c^m \ge 0$ under its corresponding majority quota policy). Suppose $S_1 \equiv \{s_1,s_2,\dots,s_{j-1}\} \subset S$ is a set of students such that there exists no school $c \in C^r$ listed by any students $s \in S_1$. Also, fixed a student $s_j \in S \backslash S_1$ and assume that her first $i - 1$ choices  $\{c_{(1)},c_{(2)},\dots,c_{(i-1)}\}$ have no intersection with schools listed by the set of students in $S_1$. The probability that the $i^{th}$ choice of student $s_j$, $c_{(i)}$, does not overlap with any schools listed by students in $S_1$ is at least
\begin{equation} \label{eq:1}
1\;-\; \sum_{c \in C_1} \alpha_c \quad -\quad \sum_{c' \in C_2} \beta_{c'}\quad - \quad \sum_{l=1}^{i-1} x_{c_{(l)}},
\end{equation}
where $C_1$ (resp. $C_2$) is the set of schools that are listed by some majority (resp. minority) students in $S_1$, and $x_{c_{(l)}} = \alpha_{c_{(l)}}$ if $s_j \in S^M$, and $x_{c_{(l)}} = \beta_{c_{(l)}}$ if $s_j \in S^m$.

Recall the Condition (4) of Definition \ref{def:regular}, we have $\alpha_c \le r \alpha_{c'}$ and $ \beta_c \le r \beta_{c'}$, $\forall c, c' \in C$. Adding these inequalities across schools gives us
\[
n \alpha_c, \,\le\, r \sum_{c' \in C} \alpha_{c'},\qquad n \beta_c \,\le\, r \sum_{c' \in C} \beta_{c'},  \qquad\qquad \forall c \in C.
\]

Thus,
\[
\alpha_c \,\le\, \frac{r \sum_{c' \in C} \alpha_{c'}}{n} \,\le\, \frac{r}{n}, \qquad\qquad \beta_c \,\le\, \frac{r \sum_{c' \in C} \beta_{c'}}{n} \,\le\, \frac{r}{n}.
\]

Note that there are at most $\bar{R}$ reserved seats, which is also an upper bound to the number of schools with either the majority quota or its corresponding minority reserve policy. Let $t^n \in (0,1)$ be the share of majority students in market $\tilde{\Gamma}^n$, and $1-t^n$ be the share of minority students.

Since at most $t^n \lambda n k$ draws can be made by all majority students from $\mathcal{A}^n$ and  $(1-t^n) \lambda n k$ draws by all minorities from $\mathcal{B}^n$, among which at most $\bar{R} /n$ comes from schools with reserved seats, we can have at most $t^n \lambda k \bar{R}$ and  $(1-t^n) \lambda k \bar{R}$ draws by majority and minority students, respectively, from schools with reserved seats. Thus, Expression \eqref{eq:1} is bounded from below by
\[
1 \;-\; \frac{t^n \lambda k \bar{R} r}{n} \;-\; \frac{(1-t^n) \lambda k \bar{R}r}{n} \;=\; 1 - \frac{\lambda k \bar{R}r}{n}.
\]

Given our arguments in Proposition \ref{prop:equal}, the two affirmative actions will result in different matching outcomes only when some schools have excessive majority applicants and insufficient number of minority applicants. To ensure $\text{Prob} \,(\tilde{\Gamma}^{n})$ converges to one when the number of schools goes to infinity, it is enough to ensure that Algorithm 1 will not occur any acceptance or rejection procedures in Step \eqref{sec:a4c} (resp. Step \eqref{sec:a4d}) when $\tilde{\Gamma}^{n}$ has majority quotas $q^{M,n}$ (resp. minority reserve $r^{m,n}$) with a high probability. That is, $\text{Prob} \,(\tilde{\Gamma}^{n})$ is bounded from below by the probability that no two particular students will list the same school $c$ with either the majority quota $q_c^M$ or its corresponding minority reserve $r_c^m$, such that $r_c^m + q_c^M =q_c$ for all $c \in C^n$, in her preference order, which is given by
\[
\text{Prob}\, ( \tilde{\Gamma}^{n} ) \ \ge  \ \bigg( 1 - \frac{\lambda k \bar{R}r}{n} \bigg)^{\lambda k \bar{R}}.
\]

Recall Condition (4) of Definition \ref{def:regular}, there exists $ \theta > 0$, such that $\bar{R} < \theta n^a$, for any $n>0$. We can rewrite the above lower bound as
\[
\bigg( 1 - \frac{\lambda k \bar{R}r}{n} \bigg)^{\lambda k \bar{R}} \ >  \ \bigg( 1 - \frac{\lambda k \theta n^ar}{n} \bigg)^{\lambda k \theta n^a}  = \left( 1 - \frac{\lambda k \theta r}{n^{1-a}} \right)^{n^{1-a} (\lambda k\theta) n^{2a-1}} \ge (e^{-\lambda k\theta r})^{\lambda k\theta n^{2a-1}},
\]
where the last inequality follows as $(1-\frac{\beta}{x})^x \ge e^{-\beta}$, when $\beta, x > 0$. Given $a \in [0, \frac{1}{2})$ by Definition \ref{def:regular}, the term $n^{2a-1}$ converges to zero as $n \to \infty$. Thus, $(e^{-\lambda k\theta r})^{\lambda k\theta n^{2a-1}}$ converges to one as $n \to \infty$, which completes the proof of Proposition \ref{prop:large}.

\bigskip
\bibliographystyle{ecta} 


\end{document}